\documentclass[twoside]{aiml22}

\usepackage{aiml22macro}

\usepackage{graphicx}
\usepackage{amsmath}
\usepackage{amssymb}
\usepackage{mathtools}
\usepackage{paralist}
\usepackage{tikz}
\usepackage{caption}
\usepackage{subcaption}
\usetikzlibrary{positioning}

\renewcommand{\phi}{\varphi}
\newcommand{\logicfont}[1]{\textsf{#1}}
\newcommand{\TPTL}{\logicfont{TPTL}}
\newcommand{\LTL}{\logicfont{LTL}}
\newcommand{\CTL}{\logicfont{CTL}}
\newcommand{\CTLs}{\logicfont{CTL}\ensuremath{^*}}

\newcommand{\CTLlps}{\ensuremath{\CTL^*_{\mathit{lp}}}}

\newcommand{\paths}{\mathit{paths}}
\newcommand{\depth}{\mathit{depth}}

\let\oldres\restriction
\renewcommand{\restriction}{\!\!\oldres}

\newcommand{\A}{\mathcal{A}}
\newcommand{\I}{\mathcal{I}}
\renewcommand{\L}{\mathcal{L}}
\renewcommand{\P}{\mathcal{P}}
\newcommand{\R}{\mathcal{R}}
\newcommand{\W}{\mathcal{W}}
\newcommand{\TP}{\mathcal{T}}
\newcommand{\Nat}{\mathbb{N}}

\newcommand{\pred}[1]{\mathit{#1}}

\newcommand{\liff}{\leftrightarrow}
\newcommand{\temporalfont}[1]{\mathsf{#1}}
\newcommand{\At}{\temporalfont{A}}
\newcommand{\Et}{\temporalfont{E}}
\newcommand{\Ft}{\temporalfont{F}}
\newcommand{\Gt}{\temporalfont{G}}
\newcommand{\Xt}{\temporalfont{X}}
\newcommand{\Ut}{\temporalfont{U}}
\newcommand{\tViol}{\mathrm{Viol}} 
\newcommand{\Viol}[5]{\tViol_{#1, #2, #3, #4}(#5)} 
\newcommand{\Viola}[5]{\tViol^{a}_{#1, #2, #3, #4}(#5)} 
\newcommand{\Violo}[5]{\tViol^{o}_{#1, #2, #3, #4}(#5)} 
\newcommand{\RViol}[6]{\mathrm{R}\tViol_{#1, #2, #3, #4, #5}(#6)} 
\newcommand{\PViol}[6]{\mathrm{P}\tViol_{#1, #2, #3, #4, #5}(#6)} 

\newcommand{\Fm}[4]{F_{#1, #2, #3}\left(#4\right)}

\newcommand{\Om}[5]{O^{#4}_{#1, #2, #3}\left(#5\right)}
\newcommand{\Od}[5]{O_{#1, #2, #3}\left(#4 < #5\right)}

\tikzset{
    into distance/.style={dash pattern=on 20pt off 3pt on 2pt off 3pt on 2pt off 3pt on 2pt},
    every node/.style={align=center}, 
    mydot/.style={fill, circle, inner sep=2pt}
}

\def\lastname{Mellema and Dignum}

\begin{document}

\begin{frontmatter}
  \title{Linking sanctions to norms in practice}
  \author{Ren\'e Mellema}\footnote{This work was partially supported by the Wallenberg AI, Autonomous Systems and Software Program (WASP) funded by the Knut and Alice Wallenberg Foundation.}
  \address{Ume\aa\ University \\ Department of Computing Science}
  \author{Frank Dignum}
  \address{Ume\aa\ University \\ Department of Computing Science}

  \begin{abstract}
  Within social simulation, we often want agents to interact both with larger systems of norms, as well as respond to their own and other agents norm violations. However, their are currently no norm specifications that allow us to interact with all of these components. To address this issue, this paper introduces the concept of violation \emph{modalities} in \CTLs\@. These modalities do not only allow us to keep track of violations, but also allow us to define the usual deontic operators. On top of this, they give us a convenient way of linking together various different norms, and allow us to reason about norms with repeated violations and obligations. We will discuss various properties of the modalities and the deontic operators, and will also discuss some ways in which this formalization can guide an implementation of normative systems.
  \end{abstract}

 \end{frontmatter}

\section{Introduction}
For a long time, deontic logics have been used as the basis for normative specification in Computer Science~\cite{Dignum:1994Dynamic,Vazquez-Salceda:2008human,Panagiotidi:2014NormAware}, which means that there are many different formalizations to pick from. However, due to the complex nature of norms, these formalizations are not always applicable in every subfield of computer science. For example, in social simulation, which deals with normative reasoning quite often, these formalizations are hardly used, because they do not give the right kind of handles for the designer of a simulation to work with. In social simulation agents often need to react to the behaviour of other agents, while most of the existing formalizations are more focused on how an agent can select the best action under the influence of a norm. This paper hopes to address this issue by introducing a logic that explicitly represents not only obligations and prohibitions generated from norms, but also represents the norm violations.

We can also find this idea of tracking violations in Andersonian reductions~\cite{Anderson:1967Nasty}, which define the deontic operators in terms of violations. It has been shown that this can be an effective approach in dynamic and temporal logics~\cite{Dignum:1994Dynamic,Dignum:2004Meeting}. However, in earlier versions a world was either in violation or not, which made it impossible to have repeating norms or repeating violations. We hope to address this issue by introducing a violation \emph{modality} that can be made unique, and be reasoned with over time as well.

This allows us to not only reason about norms, but also link different norms together, to either install reparation norms~\cite{Panagiotidi:2014NormAware}, or to specify metanorms~\cite{Axelrod:1986Evolutionary}. The former should allow us to make it easier to specify larger policies of interacting norms, making it easier for designers of social simulations to specify more complex policies for policy modelling~\cite{ASSOCC:2021}.
The latter of these are common in the social simulation literature, but to our knowledge no system exists that allows us to easily specify them in a temporal logic. 

Earlier work on norm specification has already informed us on what the components of such norms should be~\cite{Vazquez-Salceda:2008human,Dignum:2004Meeting,Panagiotidi:2014NormAware}. These are:
\begin{inparaenum}
    \item Activation condition
    \item Deactivation condition
    \item Fulfilment/Violation condition
    \item Deadline
    \item Repair
    \item Punishment
\end{inparaenum}
What we hope to include with this paper is a way to combine these components into larger policies.

An example of a situation where this would matter is a norm regulating littering at a plaza. The activation condition for an agent would then be entering the plaza, and the deactivation would be leaving the plaza. In this case, we would have a violation condition, which would be the agent throwing garbage on the ground. Since there is a violation condition, there is no deadline needed. The repair would be picking up the garbage, but having a punishment if the repair happens might not be necessary. This is where we can use a second norm to give a fine only if the repair does not happen within a certain amount of time.

We can do this by setting up a second norm, that activates on a violation of the first norm. The deactivation would be the reparation of the first norm. Now we get a norm with a fullfilment condition, which would be picking up the trash again, which can have a deadline of several timesteps. Now the punishment would be a monetary fine.


This paper is structured as follows. We will start with a discussion on earlier work, and why that is not sufficient for our purposes in Section~\ref{sec:related_work}. Then we will introduce the logic used in Section~\ref{sec:logic}. In Section~\ref{sec:violation_modality} we will introduce the violation modality, and use this to define the usual deontic operators in Section~\ref{sec:obligation_prohibition}. Section~\ref{sec:norms} shows how this can then be put together to specify and reason with norms, by means of implementing the example mentioned before. We end with a conclusion.

\section{Related work}\label{sec:related_work}
Using temporal and dynamic logics as deontic ones goes back to~\cite{Meyer:1988different}, where the reduction proposed by~\cite{Anderson:1967Nasty} was applied to dynamic logic. This meant that a violation proposition was introduced, and an action was seen as forbidden if it would lead to a violation. This idea was expanded in~\cite{Wieringa:1989Specifying,Wieringa:1991inheritance}, where the violation predicate was expanded upon by including a label and the action. However, the exact semantics are not given, and the properties of these predicates are not investigated. Furthermore, in neither of these works was the violation directly tied to the norm itself.

The first time where the violation are directly tied to the norms is in~\cite{Dignum:1994Dynamic}, where the violation predicates are labelled with a natural number, representing the norm that was being violated. The properties of this logic were then investigated, and it is shown that by only introducing this simple labelling, the framework can deal with certain paradoxes and contrary to duty obligations. The paper also discusses how we could potentially order these violations to select the best action when all possibilities lead to a violation. While this was a step in the right direction, this system still did not allow multiple violations of the same norm. This was due to the fact that the violation predicate was not specific enough.

The idea of this reduction is again applied in~\cite{Dignum:2004Meeting}, but here the reduction is to a temporal logic, in particular \CTLs\@. In this approach, the authors were interested in studying obligations with deadlines, since those are more representative of real world obligations. The authors had the problem that the violation predicates were yet again not specific enough for practical purposes. For example, this formalization could not work with repeated obligations, such as paying rent. This was due to the definition requiring that there was never a violation if the action that was required was executed. 

In~\cite{Vazquez-Salceda:2008human} the authors discuss how norms can be applied in e-Institutions. In particular they focus on how to monitor for norm violations and how to represent this in a way that makes it easy for a computer to check if there has been a violation. What is relevant for our purposes is that they identify a few components of norms that are necessary or helpful in enforcement. These are 
\begin{inparaenum} 
\item the norm condition including deadlines and activation, 
\item the violation condition, 
\item a detection mechanism, 
\item the sanction, and 
\item the repair.
\end{inparaenum} 
A lot of these components are very similar in nature, but they help the mechanisms in detecting norm violations. However, here the authors are only focused on applying norms within e-Institutions, which means that their notes for implementation are limited to that domain. Also, their focus lies solely on norms from the institutional perspective, which means that it cannot be used within an agents reasoning.

This work of norm monitoring is also discussed in~\cite{Alvarez-Napagao:2011Normative,Gomez-Sebastia:2012runtime,Panagiotidi:2014NormAware}. Like~\cite{Vazquez-Salceda:2008human}, these works are focused on the enforcement of norms in software. They have a particular interest in using Business Rule Management Systems such as Drools~\cite{drools:2008} for norm management. They do this through the use of \emph{norm instances}, so that the parameters of the norm, the violation, and the sanctions can be linked in their logical framework. Their system can also deal with norm change~\cite{Gomez-Sebastia:2012runtime}, and be applied to planning~\cite{Panagiotidi:2014NormAware}. 

\cite{Panagiotidi:2014NormAware} also discusses the norm life cycle. However, their approach with norm instances has the downside that a violated norm cannot go back to an active state. This is a limitation of the logic used, \LTL\@. They solve this by creating a new norm instance once a norm is violated. What is interesting about their planner is that is does not allow for plans that do not repair violations, but plans that violate norms are allowed. However, what is missing from their formalization is that it can only represent norms with obligations, in particular those with a maintenance condition. 


\section{\texorpdfstring{\CTLlps with extensions}{CTL*lp with extenstions}}\label{sec:logic}
As discussed before, in order to have obligations repeat we need to be able to know which deadline goes with which violation and which action (or failure to act). To facilitate this, we will need to be able to reason about the past, as well as the future. In simulations, we can be certain about the history of the run. This means that we can constrain ourselves to logics with a linear past, but a branching future, such as \CTLlps~\cite{Kupferman:2012Once,Bozzelli:2008Complexity}. 

We will also need to be able to count in order to make this work. For this, we will be borrowing some definitions from~\cite{Laroussinie:2013Counting}. However, we will not need to bring in all of their machinery, since we will not have to express how often something happened in our formulas, only in the definitions of the modalities.

Furthermore, we need to be able to distinguish between different violations. This is because we can now have multiple violations for the same norm by the same agent for the same action. One easy label that we can add is the time at which the violation occurred, which would make the violations unique. However, for this we would need some notion of a clock, or the current time. In order to incorporate this, we will be using the idea of temporal variables as used in \emph{Timed Propositional Temporal Logic} (\TPTL)~\cite{Alur:1994really}. This will both give us a notion of a clock and compare times, as well as testing whether a violation is new to this state.

Finally, as in~\cite{Dignum:2004Meeting}, we will want agents to only cause a violation if they were actually responsible for the action (not) taken. Because of this, we will also be incorporating their \emph{seeing to it that} (STIT) operator. For this, we will be incorporating some of the elements from~\cite{Dignum:2012logic}. However, since we do not consider the capabilities of agents, our definitions will be a bit simpler.

The new addition to this logic will be the addition of norm violation conditions, deadlines, and repair and punish conditions into the logic. This will allow us to define multiple violation modalities, for both acting and omitting to act, as well as modalities for repaired and punished violations. All of these modalities can occur multiple times on a run without causing problems, and will allow us to link norms to violations, and link norms to other norms.

For now, we will ignore the new violation modalities, which means that we end up with the following syntax for the language $\L$.

\begin{definition}
Given a set of propositional variables $p \in\P$, a set of agents $a \in \A$, and a set of temporal variables $t \in \TP$, and $d \in \{ +, -\}$ the language of normative specification is defined as:
\begin{align*}
    \pi :=\ & t \mid t + c \mid t - c\\
    \phi :=\ & p \mid \pi_1 < \pi_2 \mid \pi_1 = \pi_2 \mid \lnot \phi \mid \phi \land \phi \mid E_a \phi \mid \Et\alpha \mid \At\alpha \mid t.\phi(t) \\
    \alpha :=\ & \phi \mid \lnot \alpha \mid \alpha \land \alpha \mid\Xt^d \alpha \mid \Ft^d \alpha \mid \Gt^d \alpha \mid \alpha \Ut^d \alpha
\end{align*}
\end{definition}

We will be using the usual abbreviations for disjunction and implication, as well as abbreviations for the other binary operators for comparing times.

In this definition, the temporal operators are labelled with a direction, $+$ for forwards, and $-$ for backwards, but otherwise they have their traditional meaning. $t.\phi(t)$ is the temporal freeze operator from \TPTL, which binds the time of the current time point to the variable $t$, and then evaluates the formula $\phi$ using that assignment. $E_a \phi$ is the STIT operator. 

In order to interpret these formulas, we will now need a structure that explicitly incorporates time, and that shows what agents have control over which transitions. Furthermore, we need to make sure that time only branches forward, not backward. We will do this by making sure that the temporal relation over the worlds is a tree. This also allows us to use the depth of a given world in the tree as the time point for that world. Using the depth of tree to give the time point does mean that we assume that every time step takes the same amount of time, but is otherwise the same as in~\cite{Alur:1994really}.

Besides this, we will also have to encode some information about the norms into our models. This is because we will want to encode in which worlds an agent is in violation of a norm, similarly to~\cite{Anderson:1967Nasty,Dignum:1994Dynamic,Dignum:2004Meeting}. However, instead of directly marking worlds as in violation of a norm, or as a deadline for an agent, we will assign a formula to each world-norm pair, that can also depend on the agent. This requires that we also have a set $\I$ of norm labels.

\begin{definition} A model is a tuple $M = \langle\W, w, \R, T, \pi, V, D, R, P\rangle$ where:
\begin{itemize}
    \item $\W$ is a non-empty set of world states;
    \item $w \in \W$ is the root of the tree, the initial state of the simulation;
    \item $\R: \W \times \W$ is a set of transitions between world states, such that $\langle \W, \R\rangle$ is a tree;
    \item $T: \R \to 2^\A$ is a set of agent labels on elements of $\R$;
    \item $\pi: \W \times 2^\P$ is a valuation function over the world states;
    \item $V: \I \times \W \times \A \to 2^\W$ indicates in which states an agent is in violation according to a norm;
    \item $D: \I \times \W \times \A \to 2^\W$ indicates in which states an agent has a deadline for a norm;
    \item $R: \I \times \W \times \A \to 2^\W$ indicates in which states an agent has repaired a norm;
    \item $P: \I \times \W \times \A \to 2^\W$ indicates in which states an agent has been punished for a norm violation.
\end{itemize}
\end{definition}

The depth of a state in the tree is denoted with $\depth(\W, \R, s)$. The depth of the root node $w$ is 0.

We will also need the transition labels for the different agents, in order to define the semantics for $E_a$. This is defined as follows:
\begin{gather*}
    T_{as} = \{(s, s') \mid s\R s' \text{ and } a \in T((s, s'))\}
\end{gather*}

As is usual in \CTLs\ path formula are evaluated on possibly infinite paths. However, unlike traditional \CTLs\ approaches, these paths are always coming from the root of the tree. We write $\paths(\W, \R)$ for all paths, and $\paths(\W, \R, s)$ for all paths going through the state $s$. if we have a path $\sigma \in \paths(\W, \R)$, then the $j$-th state on that path is written as $\sigma(j)$. In particular, this means that $\sigma(0) = w$ for any $\sigma$, and that for all $\sigma \in \paths(\W, \R, s)$, $\sigma(\depth(\W, \R, s)) = s$, making it easy to find a particular world state in the past based on a temporal variable. Similarly, $\sigma\restriction_i$ is the finite path consisting of all elements up and to including $i$. This restriction operator will be used to make it possible to only reason over the past, and not the future.

We will also need an assignment to assign the temporal variables to time points. We will do this with the function $\tau: \TP \to \Nat$.

Over these structures, the interpretation of formulas is otherwise as standard, and the are given in Definition~\ref{def:semantics_ctl}.

Now, the last thing we need is to be able to count occurrences of a formula on a path. However, since we are now not interested in the whole path, we will adapt the definition from~\cite{Laroussinie:2013Counting} to also slice the path itself. This gives us the following definition:

\begin{definition} 
    For every path $\sigma$, the number of states satisfying a path formula $\alpha$ between depth $i$ and $j$ is noted as ${|\sigma|}^{i:j}_{\alpha}$ and defined as $|\{k \mid i \leq k \leq j \text{ and } M, \sigma\restriction_j, k \models_\tau \alpha\}|$.
\end{definition}

\section{Violation as a modality}\label{sec:violation_modality}
With all of this machinery in place, we can now extend our language $\L$ with the violation modalities. These modalities will be added as state formulas. First, we will need to add the information of when a world is in violation, passed a deadline, or had a repair or punishment happen. This is based on the $V$, $D$, $R$, and $P$ parts of the model.

\begin{definition}
    The semantics of the special propositions $V_{i,a}$, $D_{i,a}$, $R_{i,a}$, and $P_{i,a}$ are:
    \begin{align*}
        M, s \models_{\tau} V_{i, a}\ &\iff\ s \in V(i, a)\\
        M, s \models_{\tau} D_{i, a} &\iff\ s \in D(i, a)\\
        M, s \models_{\tau} R_{i, a} &\iff\ s \in R(i, a)\\
        M, s \models_{\tau} P_{i, a} &\iff\ s \in P(i, a)\\
    \end{align*}
\end{definition}

Using these, we can now introduce the violation predicates, written as $\Viola{i}{a}{t_b}{t_v}{\phi}$. This modality is supposed to be read as ``agent $a$ has violated norm $i$ at time $t_v$ by seeing to it that $\phi$, which counts as a violation since $t_b$''. Because we are working with a STIT logic however, there is a difference between acting and failing to act, so we will need two of these operators, one for causing a violation by acting, and one for causing a violation by omitting to act. These are written as $\tViol^a$ and $\tViol^o$ respectively.

The omission violation operator $\tViol^o$ also uses the parameter $t_b$ in a second way. This is because an agent might have some time before they need to see to it that $\phi$. In this case, only checking the last world before $t$ is not enough, and we need to check back until $t_b$. Later on $t_b$ will be set to the time when an obligation activated.

\begin{definition}\label{def:violation}
\begin{align*}
    M, s \models_{\tau} \Viola{i}{a}{t_b}{t_v}{\phi} \iff\ & M, s \models_{\tau}(t.(t_v = t \land V_{i, a} \land \At\Xt^-(E_a \phi)) \lor\\
    &\phantom{M, s \models_{\tau}}\At \Ft^- (t.(t_v = t) \land V_{i,a} \land \At\Xt^- (E_a \phi)) \text{ and } \\
    & M, s \models_{\tau} \At \Gt^-t.((t \geq t_b \land t \leq t_v) \to (\phi \to V_{i, a}))\\
    M, s \models_{\tau} \Violo{i}{a}{t_b}{t_v}{\phi} \iff\ & M, s \models_{\tau}(t.(t_v = t \land V_{i, a} \land D_{i, a}) \lor\\
    &\phantom{M, s \models_{\tau}}\At \Ft^- (t.(t_v = t) \land V_{i,a} \land D_{i, a}), \\
    & M, s \models_{\tau} \At \Gt^-t.((t \geq t_b \land t \leq t_v) \to ( V_{i, a} \to \lnot\phi)),\\
    & M, s \models_{\tau} \At (\lnot E_a \phi \land \lnot D_{i, a}) \Ut^- (D_{i, a} \lor t.(t_b = t)) \\
    & \text{and for all } \sigma \in \paths(\W, \R, s), \\
    & {|\sigma|}^{t_b:t_v}_{D_{i, a}} > {|\sigma|}^{t_b:t_v}_{t.(\Violo{i}{a}{t_b}{t}{\phi}) \lor E_a \phi}\\
\end{align*}
\end{definition}

\begin{figure}[hbt]
\centering
\begin{subfigure}{\linewidth}
\begin{tikzpicture}
    \node[mydot, label={left:$w$}] (w) {};
    \node[above right=of w] (i1) {};
    \node[mydot, right=of w, label={below:$t_b$}] (i2) {};
    \node[above right=of i2] (i3) {};
    \node[mydot, right=of i2] (i4) {};
    \node[above right=of i4] (i5) {};
    \node[mydot, right=of i4] (i6) {};
    \node[above right=of i6] (i7) {};
    \node[mydot, right=of i6] (i8) {};
    \node[above right=of i8] (i9) {};
    \node[mydot, right=of i8, label={below:$E_a \phi$}] (i10) {};
    \node[above right=of i10] (i11) {};
    \node[mydot, right=of i10, label={below:$V_{i, a}, \phi,t_v$\\$\Viola{i}{a}{t_b}{t_v}{\phi}$}] (i12) {};
    \node[mydot, above right=of i12, label={right:$\Viola{i}{a}{t_b}{t_v}{\phi}$}] (i14) {};
    \node[mydot, right=of i12, label={right:$\Viola{i}{a}{t_b}{t_v}{\phi}$}] (i13) {};
    \draw[into distance] (w) -- (i1);
    \draw[->] (w) -- (i2);
    \draw[into distance] (i2) -- (i3);
    \draw[->, thick, dashed] (i2) -- (i4);
    \draw[into distance] (i4) -- (i5);
    \draw[->, thick, dashed] (i4) -- (i6);
    \draw[into distance] (i6) -- (i7);
    \draw[->, thick, dashed] (i6) -- (i8);
    \draw[into distance] (i8) -- (i9);
    \draw[->, thick, dashed] (i8) -- (i10);
    \draw[into distance] (i10) -- (i11);
    \draw[->, thick, dashed] (i10) -- (i12);
    \draw[->] (i12) -- (i13);
    \draw[->] (i12) -- (i14);
\end{tikzpicture}
\caption{Example for the semantics of the violation to act modality.}\label{sfig:viol_act}
\end{subfigure}

\begin{subfigure}{\linewidth}
\begin{tikzpicture}
    \node[mydot, label={left:$w$}] (w) {};
    \node[above right=of w] (i1) {};
    \node[mydot, right=of w, label={below:$t_b$}] (i2) {};
    \node[above right=of i2] (i3) {};
    \node[mydot, right=of i2, label={below:$E_a \phi$}] (i4) {};
    \node[above right=of i4] (i5) {};
    \node[mydot, right=of i4, label={below:$D_{i,a}, \phi$}] (i6) {};
    \node[above right=of i6] (i7) {};
    \node[mydot, right=of i6] (i8) {};
    \node[above right=of i8] (i9) {};
    \node[mydot, right=of i8] (i10) {};
    \node[above right=of i10] (i11) {};
    \node[mydot, right=of i10, label={below:$D_{i,a},V_{i, a}$\\$\lnot\phi, t_v$\\$\Violo{i}{a}{t_b}{t_v}{\phi}$}] (i12) {};
    \node[mydot, above right=of i12, label={right:$\Violo{i}{a}{t_b}{t_v}{\phi}$}] (i14) {};
    \node[mydot, right=of i12, label={right:$\Violo{i}{a}{t_b}{t_v}{\phi}$}] (i13) {};
    \draw[into distance] (w) -- (i1);
    \draw[->] (w) -- (i2);
    \draw[into distance] (i2) -- (i3);
    \draw[->, thick, dashed] (i2) -- (i4);
    \draw[into distance] (i4) -- (i5);
    \draw[->, thick, dashed] (i4) -- (i6);
    \draw[into distance] (i6) -- (i7);
    \draw[->, thick, dashed] (i6) -- (i8);
    \draw[into distance] (i8) -- (i9);
    \draw[->, thick, dashed] (i8) -- (i10);
    \draw[into distance] (i10) -- (i11);
    \draw[->, thick, dashed] (i10) -- (i12);
    \draw[->] (i12) -- (i13);
    \draw[->] (i12) -- (i14);
\end{tikzpicture}
\caption{Example for the semantics of the violation for failure of acting modality.}\label{sfig:viol_omis}
\end{subfigure}
\caption{Examples with Definition~\ref{def:violation}. Formulas written beneath the states indicate those are true there, formulas not listed are assumed to be false. Temporal variables given are bound in that state.}\label{fig:viol_def}
\end{figure}

To discuss how these operators work, we want to discuss some situations where the modalities are true, and some where they are false. First, we will look at the violation to act. The model for this can be see in \autoref{sfig:viol_act}. Here we see a small model, where each circle represents a state, and the arrows indicate state transitions. 

Here at the end the agent is in violation of acting $\phi$, since two states before, the agent saw to it that $\phi$. This fulfils the first condition, that there is a state, where the state before the agent saw to it that $\phi$. For the second condition, even though it says something needs to hold globally into the past, we only need to look at the states connected by dashed arrows. This is because we only need to consider states at the interval between $t_b$ and $t_v$, since only those make the antecedent of the implication true. The check to make for these worlds is that whenever $\phi$ is true, $V_{i,a}$ is also true. $\phi$ is true in one world, and $V_{i,a}$ is also true there, Putting this together, we get that the agent is indeed in violation here.

For the second definition, we look at \autoref{sfig:viol_omis}. We again have a violation at the end, this time from omission to act. This means that there must be a point in the past where both the deadline was met, and the agent was in violation of the norm. This is the time point before, where we can indeed see that both $D_{i,a}$ and $V_{i,a}$ are true. As before, the second line again says that in all the worlds connected by the dotted arrows, $V_{i,a} \to \lnot\phi$. The direction of the implication is now reversed, since the agent does not need to see to it that always $\phi$, but only once per deadline. $V_{i,a}$ is only true once, and there we get $\phi$, so the implication holds.

For the last two conditions, the first says that the agent must not have seen to it that $\phi$ since the last deadline or $t_b$, whichever came last. As can be seen in the figure, there was no point where $E_a\phi$ was true in between the deadlines. For the last condition, we need to count how often the deadline was reached along the dashed path, and compare this to how often the agent saw to it that $E_a\phi$, or there already was a violation. There were two deadlines along the dashed path, no violations, and only one $E_a\phi$, so the agent was in violation.

This also prevented the agent from being in violation after the first deadline, since there the agent had seen to it that $\phi$, so then the last two conditions were not true. If the agent had seen to it that $E_a\phi$ twice before the first deadline, there would also not have been a violation, even though the third condition would have been true. 

However, in order to give a full norm specification, we also need to know whether a norm violation is repaired/punished. This is required for the activation/deactivation condition for the reparation/punishment norms and for meta-norms. Luckily, now defining these is relatively simple, since we can check if there have been more repair/punishment actions than unpaired violations over a given path. For this it often does not matter whether the violation came from a violation to act or a violation from omission, so we will also introduce an abbreviation for a violation in general.

\begin{definition}\label{def:repair}
\begin{align*}
    \Viol{i}{a}{t_b}{t_v}{\phi} :=\ & \Viola{i}{a}{t_b}{t_v}{\phi} \lor \Violo{i}{a}{t_b}{t_v}{\phi} \\
    M, s \models_{\tau} \RViol{i}{a}{t_b}{t_v}{t_r}{\phi} \iff\ & M, s \models_{\tau} \Viol{i}{a}{t_b}{t_v}{\phi}\\
    & \text{and } M, s \models_{\tau} \At \Ft^- (t.(t_r = t) \land E_a R_{i, a})\\
    &\text{and for all } \sigma \in \paths(\W, \R, s), {|\sigma|}^{t_v:t_r}_{E_a R_{i, a}} \geq \\
    &{|\sigma|}^{t_b:t_v}_{t.(\Viol{i}{a}{t_b}{t}{\phi} \land \lnot \Ft^+t'.(\RViol{i}{a}{t_b}{t}{t'}{\phi})}\\
    M, s \models_{\tau} \PViol{i}{a}{t_b}{t_v}{t_r}{\phi} \iff\ & M, s \models_{\tau} \Viol{i}{a}{t_b}{t_v}{\phi}\\
    & \text{and } M, s \models_{\tau} \At \Ft^- (t.(t_r = t) \land E_a P_{i, a})\\
    &\text{and for all } \sigma \in \paths(\W, \R, s), {|\sigma|}^{t_v:t_r}_{E_a P_{i, a}} \geq \\
    &{|\sigma|}^{t_b:t_v}_{t.(\Viol{i}{a}{t_b}{t}{\phi} \land \lnot \Ft^+t'.(\PViol{i}{a}{t_b}{t}{t'}{\phi})}
\end{align*}
\end{definition}

\begin{figure}[hbt]
\begin{tikzpicture}
    \node[mydot, label={left:$w$}] (w) {};
    \node[above right=of w] (i1) {};
    \node[mydot, right=of w, label={below:$t_b$}] (i2) {};
    \node[above right=of i2] (i3) {};
    \node[mydot, right=of i2] (i4) {};
    \node[above right=of i4] (i5) {};
    \node[mydot, right=of i4] (i6) {};
    \node[above right=of i6] (i7) {};
    \node[mydot, right=of i6, label={below:$t_v$\\$\Viol{i}{a}{t_b}{t_v}{\phi}$}] (i8) {};
    \node[above right=of i8] (i9) {};
    \node[mydot, right=of i8] (i10) {};
    \node[above right=of i10] (i11) {};
    \node[mydot, right=of i10, label={below:$t_r$\\$E_a R_{i,a}$}] (i12) {};
    \node[mydot, above right=of i12, label={right:$\RViol{i}{a}{t_b}{t_v}{t_r}{\phi}$}] (i14) {};
    \node[mydot, right=of i12, label={right:$\RViol{i}{a}{t_b}{t_v}{t_r}{\phi}$}] (i13) {};
    \draw[into distance] (w) -- (i1);
    \draw[->] (w) -- (i2);
    \draw[into distance] (i2) -- (i3);
    \draw[->, thick, dashed] (i2) -- (i4);
    \draw[into distance] (i4) -- (i5);
    \draw[->, thick, dashed] (i4) -- (i6);
    \draw[into distance] (i6) -- (i7);
    \draw[->, thick, dashed] (i6) -- (i8);
    \draw[into distance] (i8) -- (i9);
    \draw[->, thick, dotted] (i8) -- (i10);
    \draw[into distance] (i10) -- (i11);
    \draw[->, thick, dotted] (i10) -- (i12);
    \draw[->] (i12) -- (i13);
    \draw[->] (i12) -- (i14);
\end{tikzpicture}
\caption{Examples for Definition~\ref{def:repair}. Formulas written beneath the states indicate those are true there, formulas not listed are assumed to be false. Temporal variables given are bound in that state.}\label{fig:def_repair}
\end{figure}
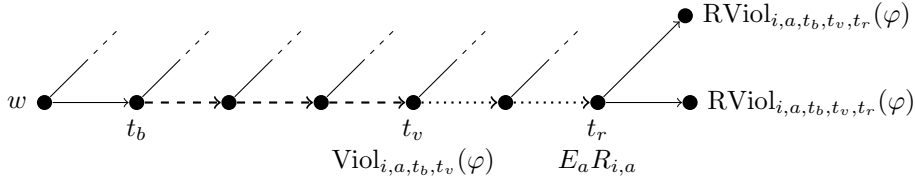

The example for this definition is in \autoref{fig:def_repair}. The example is only given for the repair, but the punishment is the same. Here we count how many repair actions have been taken along the dotted path, and compare that with the number of unrepaired violations before the violation. Here we also use the fact that the counting operator cuts off the path after the second world, so ${|\sigma|}^{t_b:t_v}$ only allows us to reason up to $t_v$, and any states later than that are not accessible. 

\subsection{Properties}
Since we now have violation as operators, one of the first questions that we can ask ourselves is whether or not these operators are normal. In our case, we do not want these operators to be normal since we do not want them to support necessitation. If they did, then every tautology would create a violation, that would be directly repaired/punished for. Luckily, since we incorporate what an agent has done into the definitions, this does not hold for any of the modalities. 
\begin{proposition}\label{prop:necessitation}
$M \models_\tau \phi$ does not imply that $M \models_\tau \Viola{i}{a}{t_b}{t_v}{\phi}$, $M \models_\tau \Violo{i}{a}{t_b}{t_v}{\phi}$, $M \models_\tau \RViol{i}{a}{t_b}{t_v}{t_r}{\phi}$, or $M \models_\tau \PViol{i}{a}{t_b}{t_v}{t_p}{\phi}$.
\end{proposition}
\begin{proof}
See appendix.
\end{proof}

One of the other properties that we do not want, is that an agent can be in violation of the same norm for the same action taken. This one is a bit hard to express if we take multiple deadlines into account, so we will focus on a short time interval. In particular, we will assume that only one deadline occurs on the interval, and that the agent can only take one action related to the norm. 
\begin{proposition}
$M, s \models_\tau (\Viola{i}{a}{t_b}{t_v}{\phi} \land \Violo{i}{a}{t_b}{t_v}{\phi}) \to \bot$
\end{proposition}
\begin{proof}
Since we have that $\Viola{i}{a}{t_b}{t_v}{\phi}$, we know that there is at least one state where the agent saw to it that $\phi$. However, this is in conflict with $\Violo{i}{a}{t_b}{t_v}{\phi}$, which specifies that the agent did not see to it that $\phi$. Therefore, we get a contradiction.
\end{proof}

Another property that is good to check is if our operator allows us to replace the formula under the operator. In particular, we don't want that all logical consequences of that formula also cause a violation. Similarly, if $\psi \to \phi$ and $\Viola{i}{a}{t_b}{t_v}{\phi}$, we do not directly want $\Viola{i}{a}{t_b}{t_v}{\psi}$. For $\Violo{i}{a}{t_b}{t_v}{\phi}$ this is a wanted property, since when $\psi$ implies $\phi$, $\psi$ would be another way to not get a violation. However, since we are treating the formulas based on their meaning, we do want to have replacement of equivalents.
\begin{proposition}
$\phi \to \psi$ and $\Viol{i}{a}{t_b}{t_v}{\phi}$ does not imply $\Viol{i}{a}{t_b}{t_v}{\psi}$.
\end{proposition}
\begin{proof}
For $\Viola{i}{a}{t_b}{t_v}{\phi}$ this follows from the fact that $\phi \to \psi$ does not imply that $\psi \to V_{i, a}$, which is a requirement for $\Viola{i}{a}{t_b}{t_v}{\psi}$ to hold. For $\Violo{i}{a}{t_b}{t_v}{\phi}$ this follows for the fact that not seeing to it that $\phi$ does not mean that the agent did not see to it that $\psi$. 
\end{proof}

\begin{proposition}
$\psi \to \phi$ and $\Viola{i}{a}{t_b}{t_v}{\phi}$ does not imply $\Viola{i}{a}{t_b}{t_v}{\psi}$.
\end{proposition}
\begin{proof}
This follows from the fact that $\psi \to \phi$ does not imply that agent $a$ saw to it that $\psi$, which is a requirement for $\Viola{i}{a}{t_b}{t_v}{\psi}$ to hold.
\end{proof}

\begin{proposition}\label{prop:violo_strengthening}
For any $s \in \W$, $M, s \models_\tau (\Violo{i}{a}{t_v}{t_v}{\phi} \land \At\Gt^-(t.(t_b \leq t \land t \leq t_v) \to (\psi \to \phi)) \to \Violo{i}{a}{t_b}{t_v}{\psi}$.
\end{proposition}
\begin{proof}
By contraposition, we know that $V_{i,a} \to \lnot\psi$. Futhermore, we know that the agent did not see to $E_a\psi$ more often than to $E_a\phi$, since $\psi \to \phi$.
\end{proof}

\begin{proposition}\label{prop:viol_replacement}
For any $s \in \W$, $M, s \models_\tau (\Viola{i}{a}{t_b}{t_v}{\phi} \land \At\Gt^-(t.(t_b \leq t \land t \leq t_v) \to (\phi \liff \psi)) \to \Viola{i}{a}{t_b}{t_v}{\psi}$ and $M, s \models_\tau (\Violo{i}{a}{t_b}{t_v}{\phi} \land \phi \liff \psi) \to \Violo{i}{a}{t_b}{t_v}{\psi}$.
\end{proposition}
\begin{proof} See appendix.
\end{proof}


Most of the interactions that the modalities have with time are fairly intuitive, and follow directly from the definitions. However, we still wanted to highlight a few. The first is that violations persist through time.
\begin{proposition}
For any model $M$ and state $s$, $M, s \models_\tau \Viol{i}{a}{t_b}{t_v}{\phi} \to \At\Gt^+\Viol{i}{a}{t_b}{t_v}{\phi}$
\end{proposition}
\begin{proof}
This follows directly from the fact that we check for violations on a path, and that the past is a tree. Any path to a world that is in the future from $s$ will also go through $s$, and thus has access in the past to the same state at depth $t_v$.
\end{proof}

Now, this does mean that we need a way to check if a violation is new. Luckily, for this we can compare the current time (which we can bind with $t.\phi(t)$) against the time the violation occurred, as such: $t.\Viol{i}{a}{t_b}{t}{\phi}$. This will only be true if the violation is introduced in the current state, since it binds the time of violation to the current time.


Since these violation modalities are modalities, it would also be interesting to see how they interact with connectives in their scope. In particular, it would be good to check under which conditions we can combine two violations, or split one into two.
\begin{proposition}\label{prop:viol_connectives}
\begin{enumerate}
\item $(\Viola{i}{a}{t_b}{t_v}{\phi} \land \Viola{i}{a}{t_b}{t_v}{\psi}) \to \Viola{i}{a}{t_b}{t_v}{\phi \land \psi}$
\item $(\Viola{i}{a}{t_b}{t_v}{\phi} \lor \Viola{i}{a}{t_b}{t_v}{\psi}) \not\to \Viola{i}{a}{t_b}{t_v}{\phi \lor \psi}$
\item $\Viola{i}{a}{t_b}{t_v}{\phi \land \psi} \not\to (\Viola{i}{a}{t_b}{t_v}{\phi} \land \Viola{i}{a}{t_b}{t_v}{\psi})$
\item $\Viola{i}{a}{t_b}{t_v}{\phi \lor \psi} \not\to (\Viola{i}{a}{t_b}{t_v}{\phi} \lor \Viola{i}{a}{t_b}{t_v}{\psi})$
\item $\Violo{i}{a}{t_b}{t_v}{\phi} \to \Violo{i}{a}{t_b}{t_v}{\phi \land \psi}$
\item $(\Violo{i}{a}{t_b}{t_v}{\phi} \lor \Violo{i}{a}{t_b}{t_v}{\psi}) \not\to \Violo{i}{a}{t_b}{t_v}{\phi \lor \psi}$
\item $\Violo{i}{a}{t_b}{t_v}{\phi \land \psi} \not\to (\Violo{i}{a}{t_b}{t_v}{\phi} \land \Violo{i}{a}{t_b}{t_v}{\psi})$
\item $\Violo{i}{a}{t_b}{t_v}{\phi \lor \psi} \to \Violo{i}{a}{t_b}{t_v}{\phi}$
\end{enumerate}
\end{proposition}
\begin{proof} See appendix. \end{proof}

Based on the proofs, we can also get a conditions under which we can do this splitting, which can help us later when deriving obligations/prohibitions from other obligations/prohibitions.

\begin{proposition}\label{prop:viola_action_specification} For any model $M$ and state $s$, 
    $M,s\models_\tau(\Viola{i}{a}{t_b}{t_v}{\phi \lor \psi} \land \At\Ft^-(E_a \phi \land \At\Xt^+(t.(t_v=t))))\to \Viola{i}{a}{t_b}{t_v}{\phi}$
\end{proposition}
\begin{proof}
    This follows from the fact that we know that $E_a \phi$ the state before the timepoint $t_v$, and the fact that $\phi \to (\phi \to \psi)$ combined with $(\phi \to \psi) \to V_{i,a}$ gives us $\phi \to V_{i,a}$.
\end{proof}


\section{Defining obligation and prohibition}\label{sec:obligation_prohibition}
Now we have our violation modalities, we can start to define the obligation and prohibition modalities. As before, the prohibition modality is easier than the obligation modality.
\begin{definition}
\begin{align*}
\MoveEqLeft
M, s\models_{\tau} \Fm{i}{a}{t_b}{\phi} \iff\\
&M, s\models_{\tau} (E_a \phi \to \At\Xt^+(t_v.\Viola{i}{a}{t_b}{t_v}{\phi}))\\
    \MoveEqLeft M, s \models_{\tau} \Od{i}{a}{t_b}{\phi}{\delta} \iff \text{for all } \sigma \in \paths(\W, \R, s), \text{ there exists }\\ 
    & j > \depth(\W, \R, s), M, \sigma(j) \models_{\tau} \delta \text{ and for all } \depth(\W, \R, s) \leq k < j:\\& M, \sigma(k) \models_{\tau} t_v.\lnot\Violo{i}{a}{t_b}{t_v}{\phi} \land \lnot\delta \text{ and if for all } \depth(\W, \R, s) \leq k < j:\\
    & M, \sigma_k \models_{\tau} \lnot E_a \phi\text{ then } M, \sigma_j \models_{\tau} t_v.\Violo{i}{a}{t_b}{t_v}{\phi}
\end{align*}
\end{definition}

The first of these says that agent $a$ is forbidden from seeing to it that $\phi$ according to norm $i$ when doing $\phi$ will lead to a violation in all next steps. 
The second definition says that an agent $a$ is obligated to do $\phi$ before $\delta$ according to norm $i$ when $\delta$ will eventually be reached, there will be no new violations for norm before $\delta$ occurs, and if the agent does not see to it that $\phi$ before $\delta$, then there will be a violation when $\delta$ occurs. Since all of these operators go over paths, we can also express this as a formula in the logic itself, just as prohibition.

Similarly as before, we can also study the properties that these operators have. As it turns out, the obligation and prohibition operators are fairly robust, and do not allow us to draw many conclusions from them.

\begin{proposition}\label{prop:split_forbidden}
\begin{enumerate}
    \item $(\Fm{i}{a}{t_b}{\phi} \land (\psi \to \phi)) \to \Fm{i}{a}{t_b}{\psi}$
    \item $\Fm{i}{a}{t_b}{\phi\lor\psi} \to (\Fm{i}{a}{t_b}{\phi} \land \Fm{i}{a}{t_b}{\psi})$
\end{enumerate}
\end{proposition}
\begin{proof}
\begin{enumerate}
    \item Take some arbitrary $M$ and $s$ such that $M, s \models_{\tau} \Fm{i}{a}{t_b}{\phi} \land (\psi \to \phi)$, and assume that $M, s \models_{\tau} E_a \psi$. This in turn means that $M, s \models_{\tau} E_a \phi$, and thus that for all $s\R s'$, $M, s \models_{\tau} t.\Viola{i}{a}{t_b}{t}{\phi}$. Since we had $M, s \models_{\tau} E_a(\psi)$ and that $M, s' \models_{\tau} \psi \to \phi$, we also get $M, s' \models_{\tau} t.\Viola{i}{a}{t_b}{t}{\psi}$, for all $s\R s'$. This shows that $M, s \models_{\tau} \Fm{i}{a}{t_b}{\psi}$.
    
    \item This follows from Proposition~\ref{prop:viola_action_specification}.
\end{enumerate}
\end{proof}

Sadly, we do not get similar behaviour for obligation. In particular $\Od{i}{a}{t_b}{\phi \land \psi}{\delta} \to (\Od{i}{a}{t_b}{\phi}{\delta} \land \Od{i}{a}{t_b}{\psi})$ would have been good to have, since this would allow us to split obligations into their component parts. However, this property does not hold due to the way in which we mark an action as causing a violation, which we also saw in Proposition~\ref{prop:viol_connectives}(vii). If we would be allowed to split a conjunction in a violation when we know that the agent had not taken that action, then this property would hold.

The obligation operator also has some weird behaviour in combination with the deadlines and repeating obligations. This is due to the fact that we do not say that something is obligatory unless there are no new violations for the new norm until the deadline. This means that we cannot have two obligations for different deadlines at the same time. In other words, if someone rents an apartment, they are not obligated to pay the rent for June until the deadline for paying the rent for May has passed, or they have payed their rent for May. This means that the framework only allows us to reason about obligations that are currently ``relevant''.

\section{Defining norms}\label{sec:norms}
Now we can use our formalization to specify norms. 
We need to define both norms for obligation and prohibition, but their general structure is the same. 

Similarly, the punishment can be both an obligation (paying a fine) and a prohibition (no longer allowed within a friend group/place). For the sake of simplicity, we will ignore the latter distinction and only define different norm types for the former. We can now see something as a norm, if after the activation condition and before the deactivation, when an agent does not see to it that the fulfilment condition is met before the deadline, they get a violation. Also, every time an agent gets a violation, they are obligated to do the repair action, and need to live with the punishment.

If we want to formalize this into our logic for the obligations, we have to keep track of all the obligations that are created after the activation, but before the deactivation. Most of these obligations begin at the deadlines, with one exception, which is the first one. Therefore, we will need two conjuncts after the activation condition, one for the first obligation, and one that deals with every occurrence of the deadline until the deactivation condition. In order to check if the norm has been violated, we can use the violation operator. Furthermore, the obligation to do the repair and the punishment only need to hold until the repair or punishment has been done (if possible).

The situation for the prohibition is easier, where we only have to see if there is a prohibition from when the norm activated until it deactivated. Handling the violations happens in a similar manner as for the norms with obligations.

\begin{definition} Norms of obligation and prohibition in a state are defined as:
\begin{align*}
    \MoveEqLeft M, s \models{_\tau} \text{NORM}^{O}(i, a, \alpha, \beta, \delta, \phi, \rho, \pi) \iff
     \\
    &M, s \models_{\tau} \alpha \to t_b.(\Od{i}{a}{t_b}{\phi}{\delta} \land \At(\delta \to \Om{i}{a}{t_b}{\phi}{\delta})\Ut^+\beta \land \\
    &\phantom{M, s \models_{\tau}}\ \At\Gt^+t_v.(\delta \land \Violo{i}{a}{t_b}{t_v}{\phi} \to (\At(O_a(\rho)\Ut^+t.(\RViol{i}{a}{t_b}{t_v}{t}{\phi}) ))) \land \\
    &\phantom{M, s \models_{\tau}}\ \At\Gt^+t_v.(\delta \land \Violo{i}{a}{t_b}{t_v}{\phi} \to (\At(O_a(\pi)\Ut^+t.(\PViol{i}{a}{t_b}{t_v}{t}{\phi}))))) \\
    \MoveEqLeft M, s \models_{\tau} \text{NORM}^{F}(i, a, \alpha, \beta, \phi, \rho, \pi) \iff \\
    &M, s \models_{\tau} \alpha \to t_b.(\At(\Fm{i}{a}{t_b}{\phi}\Ut^+\beta) \land \\
    &\phantom{M, s \models_{\tau}}\ \At\Gt^+t_v.(\delta \land \Viola{i}{a}{t_b}{t_v}{\phi} \to (\At(O_a(\rho)\Ut^+t.(\RViol{i}{a}{t_b}{t_v}{t}{\phi})))) \land \\
    &\phantom{M, s \models_{\tau}}\ \At\Gt^+t_v.(\delta \land \Viola{i}{a}{t_b}{t_v}{\phi} \to (\At(O_a(\pi)\Ut^+t.(\PViol{i}{a}{t_b}{t_v}{t}{\phi})))))
\end{align*}
\end{definition}

Notable here is that the obligations for the repair and the punishment do not need to come from the same norm. This allows us to introduce both violation handling norms~\cite{Gomez-Sebastia:2012runtime} and metanorms~\cite{Axelrod:1986Evolutionary}. In our system, the only difference between the two is for \emph{who} the norm is. A violation handling norm activates for all the agents that have \emph{caused} a violation, whereas meta-norms activate for all agents that \emph{see} a violation. 

Using this, we can revisit our example from the introduction, and formalize it using our notation. Here we will need to introduce some predicates to reason about where agents are and what effects their actions have. The first of these is $\pred{in\_plaza}(a)$, which states that agent $a$ is in the plaza. Furthermore, we also need to be able to say that agent $a$'s litter is on the ground, which we will write down as $\pred{litter}(a)$. 
\begin{align*}
    \text{NORM}^F(&i, a, \pred{in\_plaza(a)}, \lnot\pred{in\_plaza(a)}, \pred{litter(a)}, \lnot\pred{litter(a)}, \lnot\pred{litter}(a)) \\
    \text{NORM}^O(&j, a, t.\Viola{i}{a}{t_{b_i}}{t}{\pred{litter(a)}}, t.\RViol{i}{a}{t_{b_i}}{t_b}{t}{\pred{litter(a)}}, \lnot\pred{in\_plaza}(a),\\& \lnot\pred{litter(a)}, \lnot\pred{litter(a)}, \pred{pay\_fine(a)})
\end{align*}

On the reparation norm the activation and deactivation conditions can be very simple, since we can trust the logic to let us know when a norm is repaired. The only problem is that we need to somehow carry over a variable from the one norm to the other. Now we did this using the assignment function, but in an actual implementation this would be easier to manage.

Now, lets assume that in the entire simulation, we get that $\pred{throw\_can}(a) \to \pred{litter}(a)$, and that the agent is in the plaza. From the norm $i$, we can then conclude that for the agent is is forbidden to litter ($\Fm{i}{a}{t_b}{\pred{litter}(a)}$), which by Proposition~\ref{prop:split_forbidden} means that we also get $\Fm{i}{a}{t_b}{\pred{throw\_can}(a)}$. Furthermore, if the agent now decides to throw a can, we not only get a violation for the throwing of the can, but also one for the littering. This will activate the second norm $j$, and generate an obligation for the agent. No matter how often the agents throws a can (or does something else that implies littering), a new unique violation will become true, which will in turn activate the norm again.

Now we can also set up a norm that requires other agents to do something in response to the littering, which is often called a meta-norm. Meta-norms require that other agents sanction not just the breakers of a norm, but also the agents who do not sanction norm breakers themselves. 
\begin{align*}
    \text{NORM}^O(&k, b, t.\Viol{i}{a}{t_{b_i}}{t}{\pred{litter}(a)}, t.\RViol{i}{a}{t_{b_i}}{t_b}{t}{\pred{litter}(a)}, \lnot\pred{in\_plaza}(b),\\&\pred{call\_out}(b, a), \pred{call\_out}(b,a), \pred{call\_out}(b,a))\\
    \text{NORM}^O(&l, c, t.\Viol{k}{b}{t_{b_k}}{t}{\pred{call\_out}(b,a)}, t.\RViol{i}{a}{t_{b_k}}{t_b}{t}{\pred{call\_out}(b,a)}, \\&\lnot\pred{in\_plaza}(c), \pred{call\_out}(c,b), \pred{call\_out}(c,b), \pred{call\_out}(c,b))
\end{align*}

Unlike the previous norms, these norms activate and deactive when other agents break a norm. The first one, $k$, activates for agent $b$ when agent $a$ litters, and says that agent $b$ should call out agent $a$ for littering, before agent $b$ leaves the plaza. If agent $b$ fails to do that,  agent $c$ should in turn call them out, for letting the person littering get away scot free. 

These norms, in turn, also need reparation and punishment norms. However, since the action to take matches the reparation/punishment, these norms could be their own reparation/punishment norms. In order to do this, the norm should not just activate on a violation of the other norms, but also on one of the violations of the norm itself. This means that we will not need an ever escalating infinite number of norms to govern the norm following or punish behaviour of the other agents.

\section{Conclusion}
In this paper, we presented a formal framework for specifying norms, and reasoning about norms, obligations, prohibitions, and violations. As an extension on earlier work, our system can also represent repeating obligations and violations, as well as keeping track of when a norm violation is repaired/punished. 

Thanks to the properties of these violations and obligations/prohibitions, we can also go from a general specification of a norm to a more specific instantiation of an obligation/prohibition.

One of the problems in our current system, is that the way in which we mark a (lack of) action to cause a violation is still preliminary. This is due to two reasons. The first of this is that we want some notion of causality in our logic, but this is not very natural to express in temporal logics. A system where we could more directly talk about the actions an agent takes and their consequences might help us in this regard.

The second reason is that we currently use implication over time to mark certain effects of actions as causing violations. However, this means we also bring in all the traditional problems of the material conditional into our system. An alternative approach, such as using counts-as rules~\cite{Grossi:2006Classificatory} could help in this case. However, most of those are based on notions of context, and how the context should operate in a branching temporal logic is also not trivial and has to our knowledge not been investigated. We have some initial ideas on how to do this, but that requires careful investigation of the interaction between time and context.


The question of decidability and complexity of the resulting logic is also still open. While there are results for the logics that we used as a basis~\cite{Alur:1994really,Laroussinie:2013Counting,Bozzelli:2008Complexity}, no work has gone into finding out how the interactions between these components influences the resulting complexity. Similarly, an axiomatization of the system could be interesting to investigate.

Other future work that can be done based on this system is looking into how we can specify more complex normative systems using this framework. Currently, we do not link norms explicitly to one another, but we expect that this can be done. In that case, a normative system could be specified as a graph, where all the nodes are individual norms. One of the questions that are unanswered there is what kind of properties these graphs would need to get a fully enclosed normative system.

We also want to use this system to synthesise rules that future modellers can then use to help guide them in implementing norms in social simulations. We have started this by building an agent architecture for reasoning about norms and violations in social simulations, based on the rules from this framework.

\Appendix
\section{Semantics of the base logic}
\begin{definition}\label{def:semantics_ctl}
    The semantics of the logic are as follows:
    \allowdisplaybreaks
    \begin{align*}
        M, s\models_{\tau} t_1 < t_2 + c \iff& \tau(t_1) < \tau(t_2) + c\\
        M, s \models_{\tau} t_1 = t_2 + c \iff& \tau(t_1) = \tau(t_2) + c\\
        M, s \models_{\tau} t_1. \phi(t_1) \iff& M, s \models_{\tau[t_1:= \mathit{depth}(W, \R, s)]} \phi \\
        M, s \models_{\tau} E_a \phi \iff& \text{for all } (s)\R(s''), M, s'' \models_{\tau} \phi, \\
        &(s, s') \in T_{as}, \text{ and } M, w \models_{\tau} \lnot \At\Gt^+\phi\\
        M, s \models_{\tau} \Et\alpha \iff& \text{ there is a path } \sigma \in \mathit{paths}(\W, \R, s), \\
        &\text{ and }M, \sigma, \depth(\W, \R, s) \models_\tau \alpha \\
        M, s \models_{\tau} \At\alpha \iff& \text{ for all paths } \sigma \in \mathit{paths}(\W, \R, s), \\
        &\ M, \sigma, \depth(\W, \R, s) \models_\tau \alpha \\
        M, \sigma, j \models_{\tau} \phi \iff&\ M, \sigma(j) \models_{\tau} \phi \\
        M, \sigma, j \models_{\tau} t.(\alpha(t)) \iff&\ M, \sigma, j \models_{\tau[t:= j]} \alpha \\
        M, \sigma, j \models_{\tau} \Xt^+ \phi \iff&\ |\sigma| > j, M, \sigma(j + 1) \models_{\tau} \phi \\
        M, \sigma, j \models_{\tau} \Xt^- \phi \iff&\ j > 0 \text{ and } M, \sigma(j - 1) \models_{\tau} \phi \\
        M, \sigma, j \models_{\tau} \Ft^+ \phi \iff& \text{ there is an }i > j: M, \sigma(i) \models_{\tau} \phi \\
        M, \sigma, j \models_{\tau} \Ft^- \phi \iff& \text{ there is an }i < j: M, \sigma(i) \models_{\tau} \phi \\
        M, \sigma, j \models_{\tau} \Gt^+ \phi \iff& \text{ for all }i \geq j: M, \sigma(i) \models_{\tau} \phi \\
        M, \sigma, j \models_{\tau} \Gt^- \phi \iff& \text{ for all }i \leq j: M, \sigma(i) \models_{\tau} \phi \\
        M, \sigma, j \models_{\tau} \phi \Ut^+ \psi \iff& \text{ there is an } i > j \text{ such that } M, \sigma(i) \models_{\tau} \psi \\
        &\text{ and for all } j \leq k < i, M, \sigma(k) \models_{\tau} \phi\\
        M, \sigma, j \models_{\tau} \phi \Ut^- \psi \iff& \text{ there is an } i < j \text{ such that } M, \sigma(i) \models_{\tau} \psi \\
        &\text{ and for all } j \geq k > i, M, \sigma(k) \models_{\tau} \phi
    \end{align*}
\end{definition}

Conjunction and negation for both path and state formula, are defined as usual.

\section{Proofs of properties}
\begin{proof*}{Proof of Proposition~\ref{prop:necessitation}}
This follows directly from the definitions, in particular the fact that $M,s \models_\tau \phi$ can only be true if there is a world $s'$ where $M, s' \models_\tau \lnot \phi$, and for $M, s \models_\tau t.(\Violo{i}{a}{t_b}{t}{\phi})$ to be true, $M, s'' \models_\tau E_a \phi$ must be true, where $s''$ come one step before $s$. The same holds for all the other modalities.
\end{proof*}

\begin{proof*}{Proof for Proposition~\ref{prop:viol_replacement}}
We will show the case for $\Violo{i}{a}{t_b}{t_v}{\phi}$, but the other is a simplification with some minor alterations. Assume that we have some $M$, $s$, such that $M, s \models_\tau \Viola{i}{a}{t_b}{t_v}{\phi} \land \phi \liff \psi$. Now we need to show that $M, s \models_\tau \Violo{i}{a}{t_b}{t_v}{\psi}$. To show this, we only need to show that $M, s \models_{\tau} \At \Gt^-t.((t \geq t_b \land t \leq t_v) \to (\lnot\psi \to V_{i, a}))$ and that $M, s \models_{\tau} \At (\lnot E_a \psi \land \lnot D_{i, a}) \Ut^- D_{i, a}$, since in the other parts of the definition, $\psi$ is not mentioned. Both of these follow from the fact that the same formulas hold for $\phi$, and that $M, s \models_\tau \At\Gt^-t.(t_b \leq t \land t \leq t_v) \to (\phi \liff \psi)$. 
\end{proof*}

\begin{proof*}{Proof sketches for Proposition~\ref{prop:viol_connectives}}
\begin{enumerate}
    \item At the same time, so $E_a (\phi \land \psi)$, and $\phi\land\psi \to \phi$, so also $(\phi\land\psi) \to V_{i,a}$
    \item Since only either of the two has to be true, we only get that $(\phi \to V_{i,a}) \lor (\psi \to V_{i,a})$, which is not enough to conclude $(\phi\lor\psi) \to V_{i,a}$, as required for $\Viola{i}{a}{t_b}{t_v}{\phi \lor \psi}$
    \item Here we also need $\phi \to V_{i,a}$, which we cannot conclude from $(\phi\land\psi)\to V_{i,a}$.
    \item $E_a (\phi\lor\psi)$ does not imply that $E_a\phi \lor E_a\psi$, which would be required for this case.
    \item Follows from Proposition~\ref{prop:violo_strengthening} and $(\phi \land \psi) \to \phi$.
    \item The fact that the agent did not see to it that $E_a\phi$ or $E_a\psi$, does not mean that the agent did not see to it that $E_a(\phi \lor \psi)$.
    \item We need that $V_{i,a} \to \lnot \phi$, but we only know that $V_{i,a} \to \lnot(\phi \land \psi)$, and we do not have that $\lnot(\phi \land \psi) \to \lnot \phi$. Furthermore, the fact that the agent did not see to $\phi \land \psi$ does not mean that the agent did not see to them separately.
    \item Follows from Proposition~\ref{prop:violo_strengthening} and $\phi \to (\phi \lor \psi)$.
\end{enumerate}
\end{proof*}

\bibliographystyle{aiml22}
\bibliography{aiml22}

\end{document}